\documentclass[submission,creativecommons,noderivs]{eptcs}
\providecommand{\event}{SSV 2012} % Name of the event you are submitting to

\usepackage{amsmath,amssymb}
\usepackage{stmaryrd,bbm}
\usepackage{listings}
\usepackage{graphicx}
\usepackage[tight]{subfigure}

\def\cpp{C\nobreak\raise.15ex\hbox{\kern.1ex+\kern.1ex+}}

\newcommand{\std}[1]{\cite[{#1}]{cpp11}}

\newcommand{\eg}{e.g.,\ }
\newcommand{\ie}{i.e.,\ }

\newcommand{\titletext}{On the Use of Underspecified Data-Type
  Semantics \\
  for Type Safety in Low-Level Code}

\hypersetup{pdftex,
            a4paper,
            bookmarks,
            bookmarksdepth=subsection,
            colorlinks=false,
            pdfauthor={Hendrik Tews, Marcus V{\"o}lp, and Tjark Weber},
            pdftitle={\titletext},
            pdfkeywords={Data-Type Semantics, Type Safety, Formal Verification, Operating Systems}}

\usepackage{amsthm}
\usepackage{cleveref}  % must be loaded *after* hyperref

\newtheorem{lemma}{Lemma}
{\theoremstyle{definition}
 \newtheorem{definition}[lemma]{Definition}
}

\crefname{section}{Sect.}{Sections}
\crefname{figure}{Fig.}{Figures}
\crefname{table}{Tab.}{Tables}
\crefname{definition}{Def.}{Defs.}
\crefname{lemma}{Lemma}{Lemmas}
\crefname{theorem}{Theorem}{Theorems}
\crefname{equation}{eq.}{equations}
\crefname{lstlisting}{Line}{Lines}
\crefname{error}{Class}{Classes}
\crefname{appendix}{App.}{Appendix}

\Crefname{section}{Sect.}{Sections}
\Crefname{figure}{Figure}{Figures}
\Crefname{table}{Table}{Tables}
\Crefname{definition}{Definition}{Definitions}
\Crefname{lemma}{Lemma}{Lemmas}
\Crefname{theorem}{Theorem}{Theorems}
\Crefname{equation}{Equation}{Equations}
\Crefname{error}{Class}{Classes}
\Crefname{appendix}{App.}{Appendix}

\lstdefinelanguage{PVS}{
  keywords={all,and,array,assuming,assumption,axiom,begin,but,cases,
      challenge,cond,nonempty_type,
      claim,conjecture,containing,corollary,datatype,else,elsif,end,
      endassuming,endcases,endif,epsilon,exists,esporting,fact,false,forall,
      formula,from,function,if,iff,implies,importing,in,lam,lambda,law,
      lemma,let,measure,not,o,obligation,of,or,plet,postulate,proposition,
      recursive,some,sublemma,then,theorem,theory,true,type,var,when,where,
      with},
  sensitive=false,
  comment=[l]{\%}
}

\lstdefinestyle{C++}{
  language=C++,
  basicstyle=\setlength\leftskip{2mm}\fontsize{8}{9}\selectfont\ttfamily,
  numberstyle=\tiny,numbersep=2pt,
  columns=fullflexible
}

\lstdefinestyle{C++inline}{
  language=C++,
  basicstyle=\ttfamily,
  columns=fullflexible
}

\lstdefinestyle{PVS}{
  language=PVS,
  basicstyle=\setlength\leftskip{2mm}\fontsize{8}{9}\selectfont\sffamily,
  numberstyle=\tiny,numbersep=2pt,
  columns=fullflexible,
  literate={AND}{{$\wedge\,$}}1 {NOT}{{$\neg\!$}}1 {OR}{{$\vee\,$}}1
  {IMPLIES}{{$\Longrightarrow\,$}}4 {/=}{{$\neq\,$}}2 {->}{{$\to\,$}}3
  {LAMBDA}{{$\lambda$}}1 {>=}{{$\geq\,$}}2 }

\lstdefinestyle{PVSinline}{language=PVS,
  basicstyle=\fontsize{8}{9}\selectfont\sffamily,
  columns=fullflexible,
  literate={AND}{{$\wedge\,$}}1 {NOT}{{$\neg\!$}}1 {OR}{{$\vee\,$}}1
  {IMPLIES}{{$\Longrightarrow\,$}}4 {/=}{{$\neq\,$}}2 {->}{{$\to\,$}}3
  {LAMBDA}{{$\lambda$}}1 {>=}{{$\geq\,$}}2
}

\def\CVSversion$#1: 1.#2 ${#2}

\def\CVSdate$#1: #2-#3-#4 #5 ${#4.#3.#2 #5}

\def\CVSauthor$#1: #2 ${#2}

\begin{document}

\title{\titletext\thanks{
This work was in part funded by the European Commission through PASR
grant 104600, by the Deutsche Forschungsgemeinschaft through the QuaOS
project, and by the Swedish Research Council.
}}
\author{Hendrik Tews \qquad\qquad Marcus V\"olp
  \institute{Technische Universit\"at Dresden, Germany}
  \email{\{tews,voelp\}@os.inf.tu-dresden.de}
  \and
  Tjark Weber
  \institute{Uppsala University, Department of IT, Sweden}
  \email{tjark.weber@it.uu.se}}

\def\titlerunning{On the Use of Underspecified Data-Type Semantics}
\def\authorrunning{H. Tews, M. V\"olp, T. Weber}

\maketitle

\begin{abstract}
% 4-sentence approach:
% \begin{enumerate}
% \item State the problem.
% \item Say why it’s an interesting problem.
% \item Say what your solution achieves.
% \item Say what follows from your solution.
% \end{enumerate}
In recent projects on operating-system verification, C and {\cpp} data
types are often formalized using a semantics that does not fully
specify the precise byte encoding of objects.  It is well-known that
such an underspecified data-type semantics can be used to detect
certain kinds of type errors.  In general, however, underspecified
data-type semantics are unsound: they assign well-defined meaning to
programs that have undefined behavior according to the C and {\cpp}
language standards.
%
%$\quad$ 
A precise characterization of the type-correctness properties
that can be enforced with underspecified data-type semantics is still
missing.  In this paper, we identify strengths and weaknesses of
underspecified data-type semantics for ensuring type safety of
low-level systems code. We prove sufficient conditions to detect
certain classes of type errors and, finally,
identify a trade-off between the complexity of underspecified
data-type semantics and their type-checking capabilities.

% \keywords{data-type semantics, type safety, formal verification,
%   operating systems}
\end{abstract}

%%%%%%%%%%%%%%%%%%%%%%%%%%%%%%%%%%%%%%%%%%%%%%%%%%%%%%%%%%%%%%%%%%%%%%%%%%%%%%%%

\section{Introduction}
\label{sec:intro}
%======================

The formalization of C with abstract-state machines by Gurevich and
Huggins~\cite{DBLP:conf/csl/GurevichH92}, Norrish's \cpp{} semantics
in HOL4~\cite{Norrish:CPP2008} and the operating-system verification
projects VFiasco~\cite{vfiasco}, l4.verified~\cite{l4_verified} and
Robin~\cite{tews08verification} all use a semantics of C or \cpp{}
data types that employs (untyped) byte sequences to encode typed
values for storing them in memory.  An underspecified, partial
function converts byte sequences back into
typed values. 

We use the term \emph{underspecified data-type semantics} to refer to
such a semantics of data types that converts between typed values and
untyped byte sequences while leaving the precise conversion functions
underspecified.  With an underspecified data-type semantics, it is
unknown during program verification which specific bytes are written
to memory.

The main ingredients of underspecified data-type semantics are two
functions --- $\mathit{to\_byte}$ and $\mathit{from\_byte}$ ---
that convert between typed values and byte sequences. The
function \textit{from\_byte} is
in general partial,
because not every byte sequence encodes a typed value.  For instance,
consider a representation of integers that uses a parity bit:
$\mathit{from\_byte}^\text{\lstinline[style=C++inline]{int}}$ would be
undefined for byte sequences with invalid parity.
% Invalid byte
% sequences are called \emph{trap representations} in the
% C~standard~\cite[\S6.2.6.1]{c11}. 

% \textbf{XXX REV 3 OLD:}
% Underspecified data-type semantics are relevant for the verification
% of low-level systems code.  This includes code that needs to maintain
% hardware-controlled data structures, \eg page directories, or that
% contains its own memory allocator.  Type and memory safety of such
% low-level code depend on its functional correctness.  They are
% undecidable in general, and cannot be enforced with a static type
% system or other fully automatic techniques at compile time.
% %
% By using underspecified data-type semantics, one can regain some
% kind of automatic type checking for well-behaved code fragments
% in the verification environment~\cite{tuch07types}. Thereby, the
% type-correctness property can be precisely tailored to the needs
% of the specific verification goals, for instance, by taking 
% assumptions about hardware-specific data types into account.

%\textbf{XXX NEW:}
Underspecified data-type semantics are relevant for the verification
of low-level systems code.  This includes code that needs to maintain
hardware-controlled data structures, \eg page directories, or that
contains its own memory allocator.  Type and memory safety of such
low-level code depend on its functional correctness and are
undecidable in general. For this reason, type safety for such
code can only be established by logical reasoning and not by a
conventional type system. As a consequence, this paper focuses on
data-type \emph{semantics} instead of improving the type system
for, \eg \cpp.

Having to establish type correctness by verification is not as
bad as it first sounds. With suitable lemmas, type correctness
can be proved automatically for those parts of the code that are
statically type correct~\cite{tuch07types}. Thereby, the
type-correctness property can be precisely tailored to the needs
of the specific verification goals, for instance, by taking 
assumptions about hardware-specific data types into account.

It has long been known that underspecified data-type semantics can
detect certain type errors during verification, and thus imply certain
type-correctness properties~\cite{vfiasco-types}.  Because the encoding
functions $\mathit{to\_byte}$ for different types~$T$ and~$U$ are a~priori unrelated,
programs are prevented from reading a $T$-encoded byte sequence with
type~$U$.  Any attempt to do so will cause the semantics to become
stuck, and program verification will fail.

However, additional assumptions, which are often necessary to verify
machine-dependent code, easily void this property.  For instance, if
one assumes that the type
\lstinline[style=C++inline]{unsigned}~\lstinline[style=C++inline]{int}
can represent all integer values from~$0$ to~$2^n-1$ on $n$-bit
architectures,
$\mathit{from\_byte}^\text{\lstinline[style=C++inline]{unsigned}
  \lstinline[style=C++inline]{int}}$ becomes total for cardinality
reasons.  Consequently, \emph{any} sequence of~$n$ bits
becomes a valid encoding of a value of this type.
% may then be read with this type.

%% Moreover, \cpp{} has types that are not \emph{trivially
%%   copyable}~\std{\S3.9}: for instance, objects that register
%% themselves in some global data structure upon construction. A bitwise
%% copy of such an object does not preserve the associated semantics (\eg
%% it would not be automatically registered in the same global data
%% structure), and must therefore not be accessed with the object's type.
%% One insight of this paper is that simple underspecified data-type
%% semantics do not enforce this restriction: thus, they are unsound for
%% non-trivially copyable data.\footnote{There are at least two other
%%   sources of unsoundness, which we will not discuss further in this
%%   paper.  First, values may have more than one representation: \eg the
%%   value~$0$ on a ones' complement machine.  This is not adequately
%%   modeled by a functional encoding.  Second, creating trap
%%   representations in memory has undefined behavior, but underspecified
%%   data-type semantics merely prevent programs from reading them.}

Despite the widespread use of underspecified data-type semantics
%(\cref{sec:semantics}) 
for the verification of systems code, a precise characterization of
the type-correctness properties that these semantics can enforce is
still missing.

In this paper, we investigate different kinds of type errors and
different variants of underspecified data-type semantics. We
provide sufficient conditions for the fact that a certain
semantics can prove the absence of certain type errors and
describe the trade-off between the complexity of the semantics
and the type errors it can detect. One key insight is that the
simple underspecified data-type semantics that we advocated
before~\cite{vfiasco-types,tews:memory_peculiarities} is only
sound for trivially copyable data~\std{\S3.9} under
strong preconditions, which are typically violated in low-level
systems code. Type correctness for non-trivially copyable data in
the sense of \cpp{} requires a rather complicated semantics that
exploits protected bits, see \cref{sec:external_state}.

The remainder of this paper is structured as follows: in the next
section, we recollect the formalization of underspecified data-type
semantics.  \Cref{sec:type_errors} describes our classes of type
errors.  In \cref{sec:type_safety}, we formally
define \emph{type sensitivity} as a new type-correctness property that
rules out these errors, and discuss the type sensitivity of three
different variants of underspecified data-type semantics.
\Cref{sec:interplay} formally proves sufficient conditions for
type sensitivity
and
\cref{sec:related_work} discusses related work.
For space reasons, 
a small case study that exemplifies our approach 
has been moved to
\cref{sec:case_study}.
Part of our results have
been formalized in the theorem prover PVS~\cite{pvs2008}.  The
corresponding sources are publicly available.\footnote{At
  \label{lab:pvs-source}%
  \url{http://os.inf.tu-dresden.de/\~voelp/sources/type_sensitive.tar.gz}}

% In this paper, we classify different kinds of type
% errors (\cref{sec:type_errors}), and prove sufficient conditions on
% underspecified data-type semantics to enforce the absence of these
% errors (\cref{sec:interplay}).  We identify a trade-off between the
% complexity of the semantics and the different kinds of type errors
% that it can detect: subtle errors require more sophisticated semantics
% (\cref{sec:type_safety}).  A small verification case study exemplifies
% the use of underspecified data-type semantics (\cref{sec:case_study}).
% We discuss related work in \cref{sec:related_work}, before offering
% our conclusions in \cref{sec:conclusions}.  Part of our results have
% been formalized in the theorem prover PVS~\cite{pvs2008}.  Our PVS
% files are publicly available.\footnote{At
%   \url{http://os.inf.tu-dresden.de/\~voelp/sources/type_sensitive.tgz}}

\section{The Power of Underspecified Data Types}
\label{sec:semantics}
%================================================

Type checking with underspecified data-type semantics is rooted in the
observation that many programming languages do not fully specify the
encoding of typed values in memory.  For instance, the programming
languages Java~\cite{java} and Go~\cite{go} leave language
implementations (compilers and interpreters) complete freedom in how
much memory they allocate, and how values are encoded in the bytes
that comprise an object in memory.  The standards of the programming
languages C~\cite{c11} and \cpp{}~\cite{cpp11} also leave encoding and
object layout (including endianness and padding) mostly unspecified,
with few restrictions: \eg object representations must have a fixed
(positive) size.

\begin{figure}[t]
\begin{center}
\subfigure[Semantics stack]{%
  \includegraphics[width=.48\textwidth]{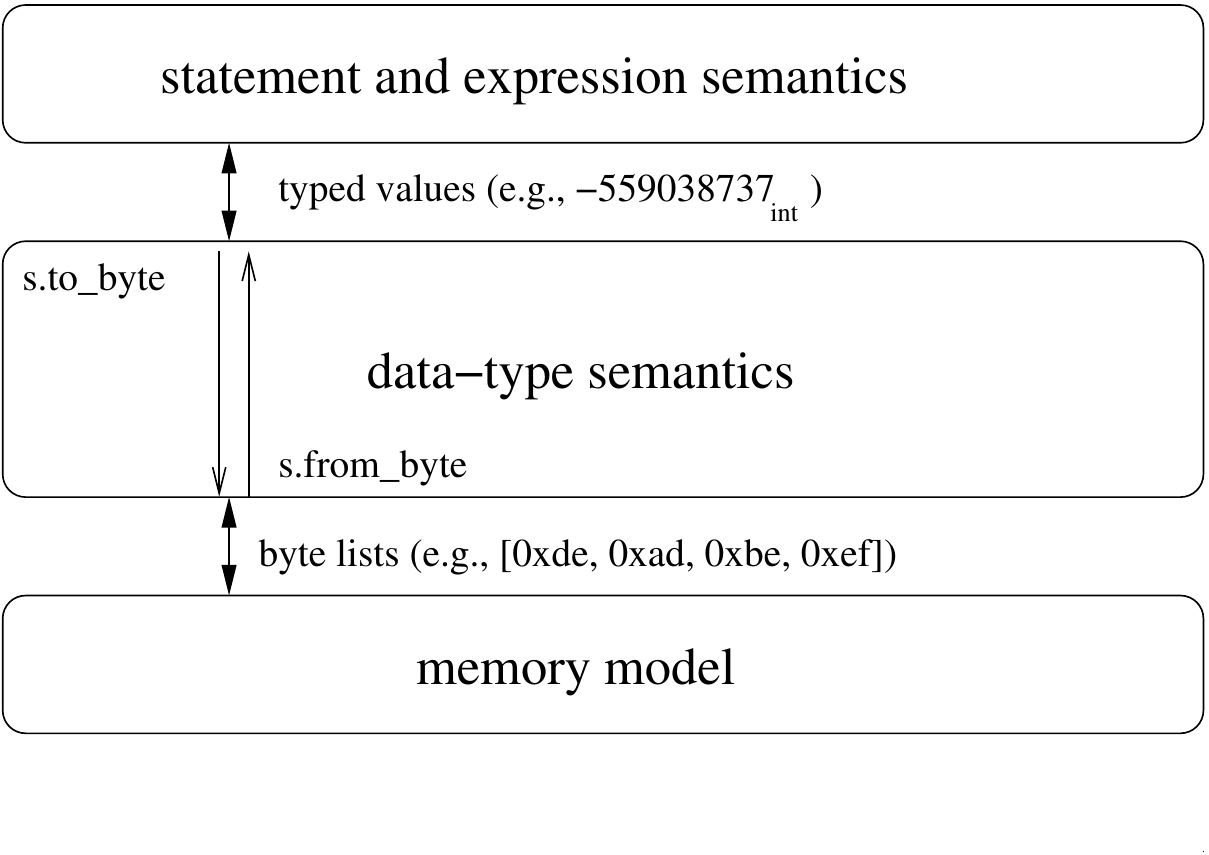}%
  \label{fig:semantics-stack}%
}\quad%
\subfigure[Approach]{%
  \includegraphics[width=.46\textwidth]{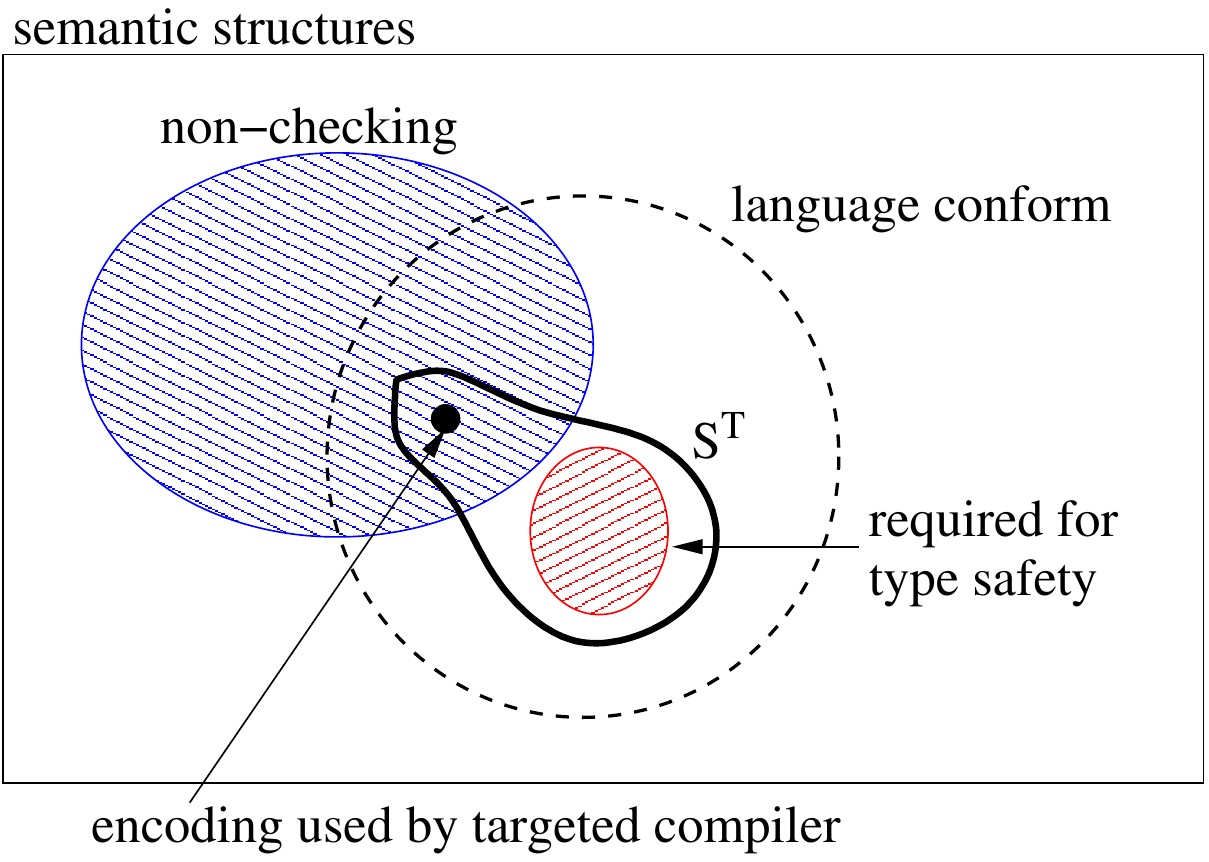}%
  \label{fig:approach}%
}%
\caption{Type checking with underspecified data-type semantics.}
\vspace{-3mm}
\label{fig:data_type_semantics}
\end{center}
\end{figure}

The data-type semantics associates a semantic structure~$s^T$ (defined
below) with each primitive language type~$T$.  Semantic structures
provide conversion functions between typed values and their memory
representation: $s^T.\mathit{to\_byte}$ translates values of type~$T$
into untyped byte lists, and $s^T.\mathit{from\_byte}$ translates byte
lists back into typed values.  The data-type semantics thereby
provides an abstract interface that connects the high-level semantics
of the language's statements and expressions to a byte-wise organized
memory model (\cref{fig:semantics-stack}).  As a beneficial side
effect, this abstraction allows us to omit from this paper both the
details of the statement/expression semantics and the details of the
memory model (which includes virtual-to-physical memory mappings and
memory-mapped devices~\cite{tews:memory_peculiarities}).

The key idea behind ensuring type safety with underspecified data-type
semantics is taken from~\cite{vfiasco-types}, and illustrated in
\cref{fig:approach}.  The conversion functions $s^T.\mathit{to\_byte}$
and $s^T.\mathit{from\_byte}$ are typically 
underspecified.\footnote{% 
  A function is underspecified if its precise mapping on values
  is not known. For an underspecified partial function the
  precise domain may also be unknown. Formally, one achieves this
  effect by using an arbitrarily chosen but fixed element of a
  suitable set of functions.
} %
The
precise requirements on these functions depend on the type~$T$ (and
possibly on additional factors such as the targeted compiler and
hardware architecture).  For instance, C requires that the type
\lstinline[style=C++inline]{unsigned}~\lstinline[style=C++inline]{char}
uses a ``pure binary'' encoding~\cite[\S6.2.6.1]{c11}.  Let
$\mathbb{S}^T$ denote the set of all semantic structures that meet
these requirements.

Program verification is carried out against an arbitrary (but fixed)
structure $s^T \in \mathbb{S}^T$, for each primitive type~$T$.
Therefore, verification succeeds only if it would succeed for
every possible choice of structures $s^T \in \mathbb{S}^T$ for
every primitive type $T$.

One can think of each structure~$s^T$ as a specific way a compiler
implements objects of type~$T$.  The set~$\mathbb{S}^T$ should contain
semantic structures that correspond to realistic compilers.  These
structures typically perform little or no runtime type checking, but
their inclusion guarantees that verification results apply to code
that is compiled and run on existing hardware.

Provided that the encoding of~$T$ is sufficiently underspecified,
$\mathbb{S}^T$ also contains more obscure semantic structures that may
not correspond to realistic compilers, but that can detect certain
type errors.  Since $s^T.\mathit{from\_byte}$ is partial, it may be
undefined for byte lists that are not of the form
$s^T.\mathit{to\_byte}(v)$ for some value~$v$.  In this case, the
semantics would get stuck when a program attempts to read an invalid
representation with type~$T$, and program verification would fail.  A
single $s^T \in \mathbb{S}^T$ whose $\mathit{from\_byte}$ function is
undefined for the invalid representation suffices to render normal
program termination unprovable.

\label{lab:s_bool}
As an example, consider the type \lstinline[style=C++inline]{bool} of
Booleans.  A semantic structure for this type is easily obtained by
imitating the encoding of values of a particular language
implementation.  For instance,
$s^\text{\lstinline[style=C++inline]{bool}}_\mathrm{gcc}$ is a
semantic structure for \lstinline[style=C++inline]{bool} if we define
$s^\text{\lstinline[style=C++inline]{bool}}_\mathrm{gcc}.\mathit{to\_byte}$
to map $\mathit{false}$ to the byte value~$[\mathtt{0x00}]$, and
$\mathit{true}$ to~$[\mathtt{0x01}]$.  The GCC C~compiler decodes
Booleans by mapping $[\mathtt{0x00}]$ to $\mathit{false}$, and all
other values to~$\mathit{true}$.  The corresponding
$s^\text{\lstinline[style=C++inline]{bool}}_\mathrm{gcc}.\mathit{from\_byte}$
function is total, and as such incapable of detecting type errors: all
byte values are valid representations
for~\lstinline[style=C++inline]{bool}.

As a second example, consider the semantic
structure~$s^\text{\lstinline[style=C++inline]{bool}}_{\mathtt{0},\mathtt{1}}$
that agrees with
$s^\text{\lstinline[style=C++inline]{bool}}_\mathrm{gcc}$, except
that~$s^\text{\lstinline[style=C++inline]{bool}}_{\mathtt{0},\mathtt{1}}.\mathit{from\_byte}$
is undefined for byte lists other than~$[\mathtt{0x00}]$
and~$[\mathtt{0x01}]$.  This structure is able to detect 
as type error
all
modifications of Boolean variables~$b$ that store a
value other than~$\mathtt{0x00}$ or~$\mathtt{0x01}$ at the address
of~$b$. When a program attempts to read~$b$ with
type~\lstinline[style=C++inline]{bool},
$s^\text{\lstinline[style=C++inline]{bool}}_{\mathtt{0},\mathtt{1}}.\mathit{from\_byte}$
is undefined for the modified value, and the semantics will get stuck.

Now suppose that $\mathbb{S}^\text{\lstinline[style=C++inline]{bool}}
= \{ s^\text{\lstinline[style=C++inline]{bool}}_\mathrm{gcc},
s^\text{\lstinline[style=C++inline]{bool}}_{\mathtt{0},\mathtt{1}},
s^\text{\lstinline[style=C++inline]{bool}}_{\mathtt{2},\mathtt{3}},
\dots \}$ additionally contains a structure
$s^\text{\lstinline[style=C++inline]{bool}}_{\mathtt{2},\mathtt{3}}$
that performs the analogous mapping for the object representations
$[\mathtt{0x02}]$~($\mathit{true}$) and
$[\mathtt{0x03}]$~($\mathit{false}$).  Because program verification is
carried out against an arbitrary (but fixed) structure
$s^\text{\lstinline[style=C++inline]{bool}} \in
\mathbb{S}^\text{\lstinline[style=C++inline]{bool}}$, programs can be
verified only if they are correct wrt.\ every structure in
$\mathbb{S}^\text{\lstinline[style=C++inline]{bool}}$.  No fixed
(constant) byte value is in the domain of all
$s^\text{\lstinline[style=C++inline]{bool}}.\mathit{from\_byte}$
functions, hence any modification that stores such a fixed value at
the address of~$b$ will be detected as a type error.  In contrast,
reading a byte list of the form
$s^\text{\lstinline[style=C++inline]{bool}}.\mathit{to\_byte}(v)$ (for
$v \in \{\mathit{true}, \mathit{false}\}$) at the address of~$b$ will
never cause an error.

For practical purposes one chooses~$\mathbb{S}^T$ as indicated by
the thick black line in \cref{fig:approach}: It should contain the set
of those semantic structures that are needed to ensure type
safety as well as those structures that represent the encoding of
the used compiler. The latter requirement ensures that the
verification results apply directly to the generated code
(assuming that the compiler is correct).

Additional assumptions typically constrain the set~$\mathbb{S}^T$ of
admissible semantic structures.  They are often necessary to prove
correctness of machine-dependent code.  For instance, assuming that
some pointer type~$T$ has the same size as
type~\lstinline[style=C++inline]{int} makes the set~$\mathbb{S}^T$
smaller, but may be required for the verification of a custom memory
allocator that casts integers into type $T$ internally.  The questions
of interest are therefore:
\begin{enumerate}
\item Which kinds of type errors can be detected with an
  underspecified data-type semantics, and
\item when do additional assumptions constrain $\mathbb{S}^T$ to a
  point where type safety is no longer guaranteed?
\end{enumerate}
Giving partial answers to these questions is the central contribution
of this paper.  We now define semantic structures more formally.

\subsection{Semantic Structures}
\label{sec:semantic_structures}
%--------------------------------

In~\cref{sec:type_safety}, we present three variants of semantic
structures with increasing type-checking capabilities.  To prepare for
the two more advanced variants, the definition below contains
``$\cdots$'' as a placeholder for further parameters.  For now, we
assume no further parameters.

\begin{definition}[Semantic structure]
  \label{def:sem-struct}
  Let $T$ be a type. A semantic structure $s = (V, A, \mathit{size},
  \mathit{to\_byte}, \mathit{from\_byte})$ for $T$ consists of a
  non-empty set of values~$V$, a set of addresses~$A \subseteq
  \mathbb{N}$ (specifying alignment requirements for objects of
  type~$T$), a positive integer $\mathit{size}$ (specifying the size
  of object encodings), and two conversion functions:
  \begin{align*}
    \mathit{to\_byte}\,\colon&\quad V\times\cdots\ \to\ list[byte]\times\cdots \\
    \mathit{from\_byte}\,\colon&\quad list[byte]\times\cdots\ \rightharpoonup\ V
  \end{align*}
  where $\mathit{list}[\mathit{byte}]$ denotes the type of byte lists,
  and $\mathit{from\_byte}$ is a partial function. Every semantic
  structure must satisfy the following properties for all $v \in V$:
  \begin{align}
    \label{eq:semantic:size}
    \mathit{length}(\mathit{to\_byte}(v, \ldots)) \quad=&\quad
    \mathit{size} \\
    \label{eq:semantic:inverse}
    \mathit{from\_byte}(\mathit{to\_byte}(v, \ldots), \ldots)
    \quad=&\quad v
  \end{align}
  where $\mathit{length}\colon \mathit{list}[\mathit{byte}] \to
  \mathbbm{N}$ denotes the length of a byte list.
\end{definition}

\Cref{eq:semantic:inverse} requires $s.\mathit{from\_byte}$ to be
defined on byte lists that form a valid object representation for~$T$.
Otherwise, $s.\mathit{from\_byte}$ may be undefined.  To ensure type
safety, we will exploit this fact by constructing sufficiently many
semantic structures whose $\mathit{from\_byte}$ function is partial:
at least one for every byte list that may have been modified by a type
error. Note that \cref{eq:semantic:inverse} ensures that one can
always read a value that has been written with the same semantic
structure. Therefore, the data-type semantics allows to verify
well-typed code as expected.

The set of values that a type can hold may depend on the hardware
architecture and compiler.  For instance, the \cpp{} type
\lstinline[style=C++inline]{unsigned int} can typically represent
values from~$0$ to~$2^n-1$ on $n$-bit architectures.  This set is
therefore specified by each semantic structure, just like size,
alignment, and encoding.  For the verification of concrete programs,
we generally assume a minimal set of values that can be represented by
all semantic structures in $\mathbb{S}^T$.

We use bytes in the definition of semantics structures, because
we assume a byte-wise organized memory model that resembles real
hardware. This is sufficient to support the bit fields of \cpp{},
because the \cpp{} standard specifies that the byte is the smallest
unit of memory modifications~\std{\S 1.7(3)--(5)}. By using lists of bits
and bit-granular addresses instead, one could support more
general architectures with more general bit fields.

\section{Type Errors}
\label{sec:type_errors}
%=======================

Type errors are undesirable behaviors of a program that result from
attempts to perform operations on values that are not of the
appropriate data type.  The causes for these errors are diverse.
Buffer overflows, dangling or wild pointers, (de-)allocation failures,
and errors in the virtual-to-physical address translation can all lead
to type errors in low-level code.

Formally, we say that memory~$m$ is \emph{modified} relative to
memory $m'$ 
at
address~$a$ if it cannot be
proven that $m$ and $m'$ are identical at $a$.
A memory modification (at address $a$) is a state transition
starting with memory $m$ and yielding memory $m'$ such that $m'$
is modified relative to $m$ (at address $a$).
A read-access to an object in memory at address $a$
with type $T$ is \emph{type correct}, if the content at~$a$ is
provably the result of a write-access to $a$ with type $T$. A
program is \emph{type correct} if all its memory-read accesses
are type correct. It turns out that the strength of
underspecified data-type semantics to detect an incorrectly typed
read access depends on the kind of memory modification that
happened before the read access in the relevant memory region. 
Therefore, we
 define \emph{type error} as a memory modification that
causes an incorrectly typed read access. To be able to view all
missing variable initialization as type error, we assume a memory
initialization that overwrites the complete memory with arbitrary
values. 

Our notion of type correctness has specific properties that are
needed for the verification of systems code. Firstly, it
permits to
arbitrarily overwrite memory whose original contents will not be
accessed any more.
Secondly, it depends on the presence and strength of additional
assumptions. If one assumes, for instance, that the range of
$to\_byte^{\text{\lstinline[style=C++inline]|void*|}}$ is contained in
the domain of
$from\_byte^{\text{\lstinline[style=C++inline]|unsigned|}}$, then
reading the value of a \lstinline[style=C++inline]|void| pointer into
an \lstinline[style=C++inline]|unsigned| variable \emph{is} type
correct.

For the analysis in the following sections, we define the following
classes of type errors (which are formally sets of memory
modifications). These classes are only used with respect to a specific
read access at some address $a$ with some semantic structure $s^T$
(for some type $T$). The memory modifications contained in some class
may therefore depend on $s^T$ and $a$.  Note that the memory
modifications in these classes may be caused not just by writing
variables in memory, but also by hardware effects, such as changes of
memory mapped registers or DMA access by external devices.

\begin{enumerate}
\item\label{error:uninitialized}{\bf Unspecified memory contents:}
  memory may contain arbitrary values; for
  instance, when the program reads a location that has not been
  initialized before. Formally, this class contains all memory
  modifications such that the modified value is indeterminate.
\item\label{error:constant}{\bf Constant byte values:} memory
  locations may contain specific (constant) byte values; for instance,
  newly allocated memory may, in some cases, be initialized to {\tt
    0x00} by the operating system or runtime environment.  
  % These byte
  % values may not constitute a valid object representation for all
  % types.
  This class contains all memory modifications where the modified
  value is a constant.
\item\label{error:implicit-casts}{\bf Object representation of a
  different type:} a $T$-typed read operation may find an object
  representation for a value of a different type~$U$ in memory. 
  Typed reads of differently typed values result in implicit casts. 
  Such an implicit cast happens, for instance, when the program attempts to
  read an inactive member of a \lstinline[style=C++inline]{union}
  type, or when a pointer of type \lstinline[style=C++inline]{T*} is
  dereferenced that actually points to an object of
  type~\lstinline[style=C++inline]{U}.

  For a given structure~$s^T$ (for some type~$T$) and an address~$a$,
  this class contains all memory modifications that write a complete
  object representation of some structure~$s^{U}$ (for a type $U \neq
  T$ with $s^U.\mathit{size} = s^T.\mathit{size}$) at~$a$.
\item\label{error:partial-read}{\bf Parts of valid object
  representation(s):} a special form of implicit cast occurs when the
  read operation accesses part of a valid object representation.  For
  instance, \lstinline[style=C++inline]{*(char*)p} reads the first
  character of the object pointed to
  by~\lstinline[style=C++inline]{p}.  A single read of a larger object
  may also span several valid object representations (or parts
  thereof) simultaneously.

  Although this error class shows many similarities to
  Class~\ref{error:implicit-casts}, it illustrates the need for
  type-safety theorems capable of ruling out undesired modifications
  at the minimal access granularity of memory (typically one byte on
  modern hardware architectures).

  For a given structure $s^T$ (for some type $T$) and an address~$a$,
  this class contains all memory modifications that overwrite the
  memory range $[a, a+s^T.\mathit{size})$ with (some slice of)
    consecutive object representations of structures~$s^{U_1}$,
    $s^{U_2}$, {\dots} (for arbitrary types~$U_i$). For $i=1$, this
    reduces to Class~\ref{error:implicit-casts}, and we require $U_1
    \neq T$---otherwise there is no type error.
\item\label{error:bitwise-copy}{\bf Bitwise copy of valid object
  representations:} objects may perform operations on construction and
  destruction; for instance, they might register themselves in some
  global data structure. A bitwise copy of such an object does not
  preserve the semantics associated with the object. Consequently, any
  attempt to access the bitwise copy with the object's type may lead
  to functional incorrectness. We consider this a type error.

  For a given structure~$s^T$ (for some type~$T$), this class contains
  all memory modifications that write at least one bit of an object
  representation of~$s^{T'}$ (for an arbitrary type~$T'$).
\end{enumerate}
We have presented these error classes in order of increasing detection
difficulty.  Class~\ref{error:bitwise-copy} is particularly
challenging, because the invalid copy is bitwise indistinguishable
from a valid object representation.  
The classes were developed during our investigation for a sound
data-type semantics for non-trivially copyable types. We make no
claim about their completeness.
In the next section, we discuss
how the different error classes may be detected with suitable variants
of underspecified data-type semantics.

%%%%%%%%%%%%%%%%%%%%%%%%%%%%%%%%%%%%%%%%%%%%%%%%%%%%%%%%%%%%%%%%%%%%%%%%%%%%%%%%

\section{Type Sensitivity with Semantic Structures}
\label{sec:type_safety}
%===================================================

In this section, we introduce the notion of \textit{type sensitivity} to
capture the requirement that no type errors occur as a result of
memory modifications. 

\begin{definition}[Type Sensitivity]
A data-type semantics for a type~$T$ is \emph{type sensitive with
  respect to a class~$\mathcal{C}$ of memory modifications} if normal
program termination implies that 
memory read with type $T$ was not changed by modifications in~$\mathcal{C}$.
%
%% no modification from~$\mathcal{C}$
%% was read with type~$T$.
\end{definition}
Applied to our approach, type sensitivity means that for every memory
modification in some class~$\mathcal{C}$, and for every subsequent read of a
modified object with type~$T$, there must be a suitable semantic
structure $s \in \mathbb{S}^T$ that can detect the modification as an
error. More precisely, $s$ is suitable if $s.\mathit{from\_byte}$ is
undefined for the modified object representation.

In the remainder of this section, we introduce three different
variants of underspecified data-type semantics:
\emph{plain object encodings}, \emph{address-dependent object encodings}, and
\emph{external-state dependent object encodings}.
%% \begin{enumerate}
%% \item \emph{plain object encodings} (i.e., no additional parameters),
%% \item \emph{address-dependent object encodings}, and
%% \item \emph{external-state dependent object encodings}.
%% \end{enumerate}
These variants are type sensitive with respect to increasingly large
classes of modifications.

\subsection{Plain Object Encodings}
\label{sec:plain_encodings}
%----------------------------

Plain object encodings for semantic structures are inspired by
trivially copyable \cpp{} data types. An object of trivially copyable
type~$T$ can be bitwise copied into a sufficiently large ($\ge
s^T.\mathit{size}$) \lstinline[style=c++inline]{char} array and back,
and to any other address holding a $T$-typed object, without affecting
its value~\std{\S3.9(2)}.  Examples of plain encodings for integers
include two's complement and sign magnitude, but also numeration
systems augmented with, \eg cyclic redundancy codes.  The semantic
structures for plain object encodings are as described in
\cref{def:sem-struct}, \ie with no additional parameters.

Plain object encodings can detect reads from uninitialized memory
(Class~\ref{error:uninitialized}) and reads of constant data
(Class~\ref{error:constant}), as exemplified in \cref{sec:semantics}.
Plain object encodings can also detect implicit casts of a differently
typed object (Class~\ref{error:implicit-casts}) or parts of it
(Class~\ref{error:partial-read}), provided $\mathbb{S}^T$ is
sufficiently rich, \ie type-sensitive with respect to the relevant
class.  We shall return to this condition in \cref{sec:interplay}.

Plain object encodings cannot detect errors from
Class~\ref{error:bitwise-copy}.

\subsection{Address-Dependent Object Encodings}
\label{sec:address_dependent}
%-----------------------------------------------

To prevent errors of Class~\ref{error:bitwise-copy}, copies by wrong
means must be detected on non-trivially copyable data
types~\std{\S3.9(2)}.  Plain object encodings cannot detect these
copies, because \cref{eq:semantic:inverse} in \cref{def:sem-struct}
requires $s.\mathit{from\_byte}$ to be defined for all byte lists that
are equal to a valid object representation.

Address-dependent object encodings are able to recognize most
(but not all) object copies obtained by bitwise memory copy
operations. For address-dependent object encodings we
augment the two conversion functions with an additional address parameter~$a$,
and adjust the left-inverse requirement of \cref{eq:semantic:inverse}
accordingly:
\begin{equation}
\label{eq:semantic:inverse_address}
  \forall a \in A.\ \mathit{from\_byte}(\mathit{to\_byte}(v, a), a)
  \ =\ v
\end{equation}
Address-dependent encodings generalize plain object encodings by
allowing a different encoding for each address.  They can therefore
detect all errors from
Classes~\ref{error:uninitialized}--\ref{error:partial-read}.  Errors
from Class~\ref{error:bitwise-copy} can be detected as long as the
bitwise copy is located at an address that is different from the
address of the original object.  This includes type errors caused by
aliasing between different virtual addresses.  However,
\cref{eq:semantic:inverse_address} prevents address-dependent object
encodings from detecting those errors of
Class~\ref{error:bitwise-copy} that overwrite memory with a bitwise
copy of an object previously stored at the same address.

\subsection{External-State Dependent Object Encodings}
\label{sec:external_state}
%------------------------------------------------------
%\clearpage

External-state dependent object encodings are the most complex
data-type semantics that we consider in this paper. They can detect type
errors from all classes discussed in \cref{sec:type_errors}, but
require further additions to the definition of semantic structures.

\subsubsection{Exploiting Protected Bits.}
\label{sec:protected_bit_encoding}
%~~~~~~~~~~~~~~~~~~~~~~~~~~~~~~~~~~~~~~~~~~~~~~~~~~~~~~

In general, error detection is easy if a part of the object
representation is protected and cannot be overwritten by
erroneous operations. One only has to make sure that the set
$\mathbb{S}^T$ contains semantic structures that store some kind
of hash in the object representation. Then, when the unprotected
part of the object representation is changed, the hash is wrong
and \textit{from\_byte} will fail. External-state dependent
object encodings develop this observation to the extreme. We will
first see that it is sufficient to protect one bit only. After
that, we will enrich the definition of semantic structures to
make sure that there is always one protected bit.

Consider a type $T$ and a set of semantic structures $\{ s^a_v
\mid a \in A, v \in V \}$ that all have the same set of values
$V$ and addresses $A$ and that all use the same object encodings,
except for the first bit. The first bit of
$s^a_v.\mathit{to\_byte}(v', a')$ is $1$ if $a = a'$ and $v = v'$
and $0$ otherwise. The function $s^a_v.\mathit{from\_byte}$ fails
if the first bit is different from what was specified for
\textit{to\_byte}. That is, every $s^a_v$ protects just the value
$v$ at address $a$ by setting the first bit of the object
representation and leaves all other value/address combinations
unprotected.

Consider now a memory copy operation that copies the object
representation of $v$ from address $a$ to a different address
$a'$ but leaves the first bit at address $a'$ intact. If this bit
is $0$ then $\mathit{from\_byte(a', \ldots)}$ from structure
$s^{a'}_v$ will fail. If the first bit at $a'$ was $1$,
$s^{a'}_v$ will succeed, but all other structures will fail. In
case $a$ and $a'$ are the same address, the memory remains
(provably) unchanged, so there is no error to detect. However, if
the value $v$ at $a'$ is overwritten with the object
representation of a value $v'$ that was previously stored there,
either $s^{a'}_{v'}$ or $s^{a'}_v$ will detect the error in case
the first bit at ${a'}$ remains unchanged.

We can conclude that a sufficiently large set~$\mathbb{S}^T$ can
detect all type errors from all classes, provided there is at least
one protected bit that no erroneous memory modification can change.

\subsubsection{Protecting Bits in External State.}
\label{sec:protecting_external_bits}
%~~~~~~~~~~~~~~~~~~~~~~~~~~~~~~~~~~~~~~~~~~~~~~~~~~~~~~

We will now enrich semantic structures such that every object
representation can potentially contain one additional bit. With a
clever use of underspecification this will require only one
additional bit of memory per program. In a last step we will
protect this one bit by making its location unknown.

We first enrich semantic structures with a partial function
\textit{protected\_bit}:
\begin{align*}
  \mathit{protected\_bit}\, \colon&\quad A \rightharpoonup \mathit{B} \\
  \mathit{to\_byte}\,\colon&\quad V \times A
  \to \mathit{list}[\mathit{byte}] \times \mathit{bit} \\
  \mathit{from\_byte}\,\colon&\quad
    \mathit{list}[\mathit{byte}] \times A \times \mathit{bit} \rightharpoonup V
\end{align*}
Here, $\mathit{B}$ is the set of bit-granular addresses of the underlying
memory and \textit{bit} is the type of bits. The idea is as follows:
If $s.\mathit{protected\_bit}(a) = \mathit{b}$ then the structure $s$
uses an additional bit of object representation at address $b$ for
values stored at address $a$. In this case \textit{to\_byte} returns
this additional bit and \textit{from\_byte} expects it as third
argument. If $\mathit{protected\_bit}(a)$ is undefined, no additional
bit is used and \textit{to\_byte} returns a dummy bit. The consistency
requirement of semantic structures (\cref{eq:semantic:inverse} in
\cref{def:sem-struct}) is changed in the obvious way. In the
verification environment (\cref{fig:semantics-stack}) the memory model
must of course be adapted to handle the additional bit appropriately.

There are of course problems if the address returned by
\textit{protected\_bit} is already in use. We solve this in
several steps. We first require that $\mathit{protected\_bit}$ is
defined for at most one address for every structure $s$. This
restriction does not hurt because $\mathbb{S}^T$ can still
contain one structure for every address~$a$ such that
$\mathit{protected\_bit}(a)$ is defined. Next, recall from
\Cref{sec:semantics} that for each primitive type $T$ a fixed but
arbitrarily chosen $s^T$ is used. We refine this choice such that
there is at most one primitive type $T$ for which
$s^T.\mathit{protected\_bit}$ is defined for one address. Again,
this latter restriction does not limit the checking powers,
because for each type every address still can
potentially use an additional bit.

As a last step consider the set $\mathit{AF}$ of free, unused bit granular
addresses.\footnote{%
  For practical purposes one can use a safe
  approximation of $\mathit{AF}$.
} %
The memory model is enriched with a constant $r \in \mathit{AF}$ that is
used precisely when the only additional bit that is used by the
current choice of semantic structures is outside of $\mathit{AF}$. In this
case, the memory model silently swaps the contents of $r$ and the
additional bit.

The changes for using protected bits are rather complex. However,
if $\mathit{AF}$ is not empty and if the sets $\mathbb{S}^T$ are
sufficiently large, then there is for each type $T$ and each
address $a$ a choice of semantic structures such that values of
type $T$ at address $a$ use an additional bit in the object
representation. If, for every
accessible bit
address~$b$, every~$\mathbb{S}^T$ contains
a structure that uses~$b$ as additional bit, then
the location of the additional bit is de~facto unknown. Under
these circumstances protected bits can detect all type errors
from all classes of \cref{sec:type_errors} as 
long as it is not the case that the complete memory is
overwritten. 

There are two points to note about external-state dependent object
encodings. Firstly, the protected bit in these encodings is not write
protected in a general sense. Type-correct operations that use the
chosen semantic structure do in fact change the protected
bit. Secondly, we used single bits and bit-granular addresses here
only because we assume a memory model that resembles real
hardware. The same idea can be applied to more abstract memory models.

%%%%%%%%%%%%%%%%%%%%%%%%%%%%%%%%%%%%%%%%%%%%%%%%%%%%%%%%%%%%%%%%%%%%%%%%%%%%%%%%

\section{Towards a Type-Sensitivity Theorem}
\label{sec:interplay}
%============================================

In \cref{sec:type_safety}, we carved out type sensitivity as the key
property that ensures there are sufficiently diverse semantic
structures to identify all type errors.  We now take a closer look at
the delicate interplay between compiler intelligence, additional
assumptions, and type sensitivity.  We give sufficient conditions for
type sensitivity for the error classes discussed in
\cref{sec:type_errors}.  These entail construction guidelines for
sufficiently rich sets~$\mathbb{S}^T$.

The relationship between semantic structures and type errors that are
ruled out by verification turns out to be intricate.  Intuitively, one
might expect type sensitivity to be monotone: more semantic structures
can detect more type errors.  Unfortunately, more semantic structures
also give rise to more program executions, and can therefore cause
undetected type errors.

For instance, consider a memory-mapped device that overwrites memory
at an address~$a$.  A program that performs a read access of type~$T$
will be unaffected by this modification if alignment requirements
ensure that objects of type~$T$ are never located at~$a$.  Relaxing
these alignment requirements, however, might lead to a type error in
certain program executions: namely in those that read at~$a$.  To
remain type sensitive, the data-type semantics would then need to
admit a semantic structure that can detect the modification \emph{and}
allows alignment at~$a$.

To be able to detect a memory modification as a type error (without
resorting to external state), we have to assume some degree of
independence between the modification and
the semantic structures fixed for program
verification.  The classification in \cref{sec:type_errors} describes
different degrees of data independence.  In this section, we
additionally assume that modifications occur at fixed addresses,
independent of the choice of semantic structure for~$T$.

\begin{figure}
%  \vspace{2mm}%manual layout
  \begin{center}
    \includegraphics[width=.48\textwidth]{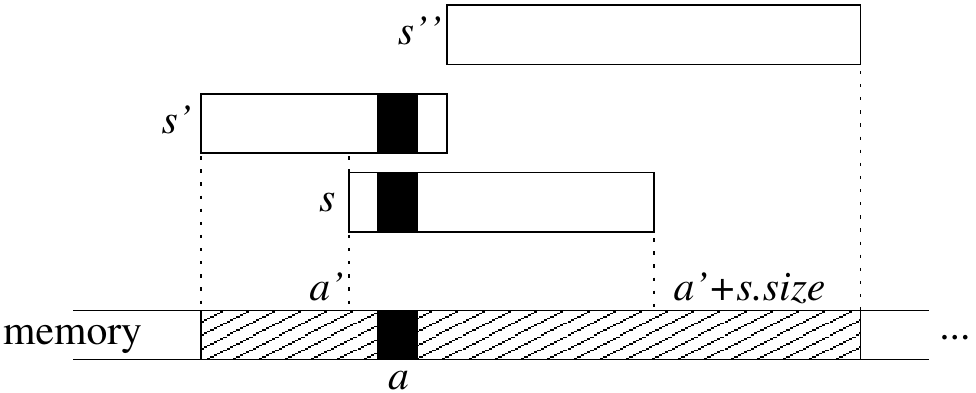}
    \caption{\label{fig:visibility}Visibility of addresses to semantic
      structures: $a$ is visible to $s$ and $s'$ but not to $s''$.}
%\vspace{-1mm}
  \end{center}
\end{figure}

We say that an address~$a$ is \emph{visible to a semantic
  structure~$s$} if there exists $a' \in s.A$ such that $a' \leq a <
a'+s.\mathit{size}$; see \Cref{fig:visibility} for illustration.  In
other words, $a$ is visible to~$s$ if $s$ might read memory at~$a$. We
say that $a$ is \emph{visible} if there is a semantic structure $s \in
\mathbb{S}^T$ such that $a$ is visible to~$s$.

\begin{lemma}[Unspecified Memory]\label{lemma:uninitialized}
Assume that for every visible address~$a$, there is a semantic
structure $s \in \mathbb{S}^T$ and an address $a' \in s.A$ (with $a'
\leq a < a'+s.\mathit{size}$) such that for every sequence of bytes
$(b_i)_{i=0}^\infty$, there is a byte value~$b$ such that
$s.\mathit{from\_byte}$ is undefined for the byte list $[b'_{a'},
  \ldots, b'_{a'+s.\mathit{size}-1}]$ given by $b'_a := b$, $b'_i :=
b_i$ for $i \neq a$.  Then $\mathbb{S}^T$ is type sensitive
wrt.\ unspecified memory contents (Class~\ref{error:uninitialized}).
\end{lemma}

\begin{proof}
Assume that an unspecified byte value at address~$a$ is read with
type~$T$.  Because $a$ is visible, there is a semantic structure~$s
\in \mathbb{S}^T$ as postulated in the lemma.  This structure might
read at address~$a'$.  Let $(b_i)_{i=0}^\infty$ be the memory contents
at the time of the read.  Since $b_a$ is unspecified, it might be
equal to~$b$.  Hence $s.\mathit{from\_byte}$ might read the byte list
$[b'_{a'}, \ldots, b'_{a'+s.\mathit{size}-1}]$, for which it is
undefined.  Therefore, normal program termination is no longer
provable.%\hfill$\qed$
\end{proof}
For instance, $\mathbb{S}^\text{\lstinline[style=C++inline]{bool}} :=
\{ s^\text{\lstinline[style=C++inline]{bool}}_{\mathtt{0}, \mathtt{1}}
\}$ (with $s^\text{\lstinline[style=C++inline]{bool}}_{\mathtt{0},
  \mathtt{1}}$ as defined on page~\pageref{lab:s_bool}) is type
sensitive wrt.\ unspecified memory contents, because
$s^\text{\lstinline[style=C++inline]{bool}}_{\mathtt{0},
  \mathtt{1}}.\mathit{from\_byte}$ is undefined for some (in fact, for
all but two) byte lists of length one.

%% Additional assumptions can easily constrain~$\mathbb{S}^T$ so that no
%% semantic structure fulfills the requirements of
%% \cref{lemma:uninitialized}.  For instance, if we assume that the
%% object representation of \lstinline[style=C++inline]{unsigned int} has
%% a size of 32~bits, and that objects of this type can hold all values
%% from $0$ to $2^{32}-1$, then for any $s \in
%% \mathbb{S}^\mathtt{unsigned\ int}$, $s.\mathit{from\_byte}$ is total
%% simply for cardinality reasons.  Fortunately, one can often relax
%% these assumptions, \eg by also allowing semantic structures where
%% \lstinline[style=C++inline]{unsigned int} occupies more than four
%% bytes.

The following lemmas have similarly straightforward proofs, which we
omit for space reasons.  For constant byte values
(Class~\ref{error:constant}), the only difference to
\cref{lemma:uninitialized} is that \emph{any} byte value~$b$ must now
be detected as an error.  In particular, any semantics that is type
sensitive wrt.\ constant byte values is also type sensitive
wrt.\ unspecified memory contents.

\begin{lemma}[Constant Bytes]\label{lemma:constant}
Assume that for every visible address~$a$, and for every byte
value~$b$, there is a semantic structure $s \in \mathbb{S}^T$ and an
address $a' \in s.A$ (with $a' \leq a < a'+s.\mathit{size}$) such that
for every sequence of bytes $(b_i)_{i=0}^\infty$,
$s.\mathit{from\_byte}$ is undefined for the byte list $[b'_{a'},
  \ldots, b'_{a'+s.\mathit{size}-1}]$ given by $b'_a := b$, $b'_i :=
b_i$ for $i \neq a$.  Then $\mathbb{S}^T$ is type sensitive
wrt.\ constant byte values (Class~\ref{error:constant}).
\end{lemma}
\Cref{sec:semantics} exemplifies how a sufficiently rich set $\mathbb{S}^T$ can be
obtained by inclusion of sufficiently many semantic structures such
that no byte list is in the domain of all $\mathit{from\_byte}$
functions.

For Class~\ref{error:implicit-casts}, we restrict ourselves to those
semantic structures $s^T \in \mathbb{S}^T$ that read exactly one
object representation produced by a semantic structure $s^U \in
\mathbb{S}^U$ for some type~$U$.  Partial overlaps between object
representations are covered by Class~\ref{error:partial-read}.  We
assume that the object representation for~$U$ does not depend on the
choice of semantic structure for~$T$.

\begin{lemma}[Implicit Casts]\label{lemma:implicit-casts}
Assume that for every semantic structure $s^T \in \mathbb{S}^T$, every
address $a \in s^T.A$, and every byte list $[u_a, \ldots,
  u_{a+s^T.\mathit{size}-1}]$ that is the result of
$s^U.\mathit{to\_byte}(v, \ldots)$ for some $s^U \in \mathbb{S}^U$, $v
\in s^U.V$, there is a semantic structure $s \in \mathbb{S}^T$ and an
address $a' \in s.A$ such that for every sequence of bytes
$(b_i)_{i=0}^\infty$, $s.\mathit{from\_byte}$ is undefined for the
byte list $[b'_{a'}, \ldots, b'_{a'+s.\mathit{size}-1}]$ given by
$b'_i := u_i$ for $a \leq i < a+s^T.\mathit{size}$, $b'_i := b_i$
otherwise.  Then $\mathbb{S}^T$ is type sensitive wrt.\ implicit casts
from type~$U$ (Class~\ref{error:implicit-casts}).
% Tweak to introduce a label for Lemma 4
%\refstepcounter{lemma}
%\label{lemma:partial-read}
\end{lemma}
To construct a set $\mathbb{S}^T$ that fulfills the assumptions of the
preceding lemma, one can include a set~$S$ of non-total semantic
structures that are closed wrt.\ permutation of undefined object
representations.  Given a non-total semantic structure $s$ where
$s.\mathit{from\_byte}(bl, \ldots)$ is undefined, we can construct
such a set $S$ if we include for all byte lists $bl'$ the semantic
structure $s'$ that is identical to $s$ except that
$s'.\mathit{from\_byte} = \Pi_{bl,bl'} \circ s.\mathit{from\_byte}$
and $s'.\mathit{to\_byte} = s.\mathit{to\_byte} \circ \Pi_{bl,
  bl'}$. Here, $\Pi_{bl, bl'}$ is the permutation function that just
exchanges $bl$ with $bl'$.

%% \begin{floatingfigure}{.3\textwidth}
%%   \begin{center}
%%     \includegraphics[width=.28\textwidth]{bitwise-copy}
%%     \caption{\label{fig:bitwise-copy} Byte-wise copy}
%%   \end{center}
%% \end{floatingfigure}

It is straightforward to generalize \cref{lemma:implicit-casts} to
parts of valid object representations (Class~\ref{error:partial-read})
by allowing $[u_a, \ldots, u_{a+s^T.\mathit{size}-1}]$ to be (an
arbitrary slice of) a concatenation of object representations for
other types~$U_i$. We omit the formal statement of this lemma.

To detect Class~\ref{error:bitwise-copy} errors, we have to further
relax our independence requirements between type errors and semantic
structures by considering also copies of object representations
for~$T$ at visible addresses~$a$.
%% the composition of read byte lists from the
%% bits of valid object representations of the read type $T$ 
%% (see \cref{fig:bitwise-copy}). 
We say that two semantic structures $s_1$ and $s_2$ are
\emph{equivalent}, $s_1 \sim s_2$, if they differ at most in their
$\mathit{to\_byte}$, $\mathit{from\_byte}$ functions. Equivalent
semantic structures produce and interpret object representations of
the same size and at the same set of addresses.

\begin{lemma}[Bitwise Copy]\label{lemma:bitwise-copy}
Assume that for every semantic structure $s \in \mathbb{S}^T$ and
every $a \in s.A$ there exists an equivalent semantic structure $s'
\in \mathbb{S}^T$ such that for any byte list $bl := [b_a, \ldots, b_{a +
s.\mathit{size} - 1}]$ where $b_i$, $i \in [a, a + s.\mathit{size})$
may be comprised of copies of bit value of an object representation
$s.\mathit{to\_byte}(v, \ldots)$ for some value $v \in s.V$, the result
of $s'.\mathit{from\_byte}(bl, \ldots)$ is undefined if we replace the
copied bits with the respective value of $s'.\mathit{to\_byte}(v,
\ldots)$.  Then $\mathbb{S}^T$ is type-sensitive wrt.\ bitwise copies of
a non-trivially copyable object (Class~\ref{error:bitwise-copy}).
\end{lemma}
Clearly, if the copy is exact in the sense that $bl =
s'.\mathit{to\_byte}(v, \ldots)$ for some value $v \in s.V$,
\cref{eq:semantic:inverse} rules out the existence of a semantic
structure $s'$ for which $s'.\mathit{from\_byte}(v, \ldots)$ is
undefined. For the same reason, there can be no semantic structure
with an address dependent encoding that detects
Class~\ref{error:bitwise-copy} errors if $bl$ is a valid object
representation for the read address $a$.

In \cref{sec:external_state}, we described external-state dependent
encodings that are able to fulfill the assumptions of
\cref{lemma:bitwise-copy}.  The proof that external-state dependent
encodings are type-sensitive wrt.\ all error classes is lengthy but
not difficult. It builds on the fact that for every address $a$
there exists a choice of semantic structures such that values at
address $a$ are protected with one additional bit of object
representation, see \cref{sec:external_state}.

%-------------------------------------------------------------------------------

\section{Related Work}
\label{sec:related_work}
%========================

In spirit, the work presented here is very similar to runtime type
checking, as it is present in dynamically typed programming languages
such as Lisp and Perl. The runtime system of such languages attaches
type tags to all values, and uses them for type checking at runtime.
There are also tools that perform extended static or dynamic type
checking for C and \cpp{} programs by source or object code
instrumentation~\cite{Burrows-run-time-check,Loginov-run-time-check}.
One can view each element of $\mathbb{S}^T$ as a runtime system that
performs a particular type check.  While runtime type checking can
practically only be done for a limited number of program runs, this
paper analyzes verification techniques that apply to all possible
program runs.

There are several proposals to enhance the type safety of~C.
Cyclone~\cite{cyclone} introduces additional data types such as safe
pointers.  BitC~\cite{bitc} augments a type-safe dialect of~C with
explicit placement and layout controls to reduce the number of
situations where low-level code has to break the type system.  A
strength of underspecified data-type semantics is the ability to
re-establish type safety when such a break is inevitable.

As mentioned in \cref{sec:intro}, several similar data-type semantics
for C or \cpp{} have been discussed in the literature. The
formalization of C with abstract state machines by Gurevich and
Huggins~\cite{DBLP:conf/csl/GurevichH92} and Norrish's \cpp{}
semantics in HOL4~\cite{Norrish:CPP2008} both rely on partial
functions to convert byte lists to typed values.

The idea to reflect the underspecification of the \cpp{} standard in
the data-type semantics, and to exploit this underspecification for
type checking, was first proposed in the context of the VFiasco
project~\cite{vfiasco} by Hohmuth and Tews~\cite{vfiasco-types}.  This
idea has then been independently further developed in the
operating-system verification projects l4.verified~\cite{l4_verified}
and Robin~\cite{tews08verification}.

For l4.verified, Tuch et~al.\ built a typed memory on top of untyped
memory~\cite{tuch07types}. This typed view on memory can be used to
automatically discharge type-correctness conditions for type-safe code
fragments.

\section{Conclusions}
\label{sec:conclusions}
%=======================

In this paper, we explored the ability of underspecified data-type
semantics to enforce type-correctness properties in verification
settings that rely on untyped byte-wise organized memory.  We have
identified five different classes of type errors, and proved
sufficient conditions for recognizing all type errors from each class.
This required increasingly complex data-type semantics.  
% \\
% \textbf{XXX REV 2, OLD:}
% Notably,
% bitwise copies of non-trivially copyable objects can only be detected
% with external-state dependent object encodings, showing that simple
% underspecified data-type semantics are unsound for such types.
% \\\textbf{XXX NEW:}
Notably,
simple
underspecified data-type semantics are unsound for non-trivially
copyable types. Bitwise copies of such types can only be detected
with external-state dependent object encodings. The trade-off
between using such complex data-type semantics or dealing with
errors from class~\ref{error:bitwise-copy} by other means must be
decided for each verification individually.

Although our analysis is inspired by C and \cpp{}, our results are
largely programming-language independent.  They apply to all programs
that cannot be statically type-checked.  To demonstrate the practical
relevance of our analysis, we verified the type safety of a small code
fragment from an operating-system kernel in PVS (see
\cref{sec:case_study}).  Our PVS files are publicly available (see
\cref{lab:pvs-source} on page~\pageref{lab:pvs-source}).

Giving a fully accurate, sound data-type semantics for the
verification of C and \cpp{} code remains a challenge.  The language
standards make few guarantees in order to permit efficient
implementations on a wide range of hardware architectures. Yet the
type systems are complex, there are subtle constraints on memory
representations and type domains, and the typed and untyped views on
memory interact in intricate ways.

%%%%%%%%%%%%%%%%%%%%%%%%%%%%%%%%%%%%%%%%%%%%%%%%%%%%%%%%%%%%%%%%%%%%%%%%%%%%%%%%

\bibliographystyle{eptcs}
\bibliography{own}

\begin{thebibliography}{10}
\providecommand{\bibitemdeclare}[2]{}
\providecommand{\surnamestart}{}
\providecommand{\surnameend}{}
\providecommand{\urlprefix}{Available at }
\providecommand{\url}[1]{\texttt{#1}}
\providecommand{\href}[2]{\texttt{#2}}
\providecommand{\urlalt}[2]{\href{#1}{#2}}
\providecommand{\doi}[1]{doi:\urlalt{http://dx.doi.org/#1}{#1}}
\providecommand{\bibinfo}[2]{#2}

\bibitemdeclare{inproceedings}{Burrows-run-time-check}
\bibitem{Burrows-run-time-check}
\bibinfo{author}{Michael \surnamestart Burrows\surnameend},
  \bibinfo{author}{Stephen~N. \surnamestart Freund\surnameend} \&
  \bibinfo{author}{Janet~L. \surnamestart Wiener\surnameend}
  (\bibinfo{year}{2003}): \emph{\bibinfo{title}{Run-Time Type Checking for
  Binary Programs}}.
\newblock In \bibinfo{editor}{G{\"o}rel \surnamestart Hedin\surnameend},
  editor: {\sl \bibinfo{booktitle}{12th International Conference on Compiler
  Construction}}, {\sl \bibinfo{series}{LNCS}} \bibinfo{volume}{2622},
  \bibinfo{publisher}{Springer}, pp. \bibinfo{pages}{90--105},
  \doi{10.1007/3-540-36579-6\_7}.

\bibitemdeclare{manual}{go}
\bibitem{go}
 (\bibinfo{year}{2012}): \emph{\bibinfo{title}{The {Go} Programming Language
  Specification}}.
\newblock \urlprefix\url{http://golang.org/doc/go_spec.html}.
\newblock \bibinfo{note}{Retrieved June~15, 2012.}

\bibitemdeclare{book}{java}
\bibitem{java}
\bibinfo{author}{James \surnamestart Gosling\surnameend}, \bibinfo{author}{Bill
  \surnamestart Joy\surnameend}, \bibinfo{author}{Guy~L. \surnamestart {Steele
  Jr.}\surnameend} \& \bibinfo{author}{Gilad \surnamestart Bracha\surnameend}
  (\bibinfo{year}{2005}): \emph{\bibinfo{title}{The {Java} Language
  Specification (3rd ed.)}}.
\newblock \bibinfo{publisher}{Addison-Wesley}.

\bibitemdeclare{article}{cyclone}
\bibitem{cyclone}
\bibinfo{author}{D.~\surnamestart Grossman\surnameend},
  \bibinfo{author}{M.~\surnamestart Hicks\surnameend},
  \bibinfo{author}{T.~\surnamestart Jim\surnameend} \&
  \bibinfo{author}{G.~\surnamestart Morrisett\surnameend}
  (\bibinfo{year}{2005}): \emph{\bibinfo{title}{{Cyclone}: A Type-Safe Dialect
  of {C}}}.
\newblock {\sl \bibinfo{journal}{C/C++ User's Journal}}
  \bibinfo{volume}{23}(\bibinfo{number}{1}).

\bibitemdeclare{inproceedings}{DBLP:conf/csl/GurevichH92}
\bibitem{DBLP:conf/csl/GurevichH92}
\bibinfo{author}{Yuri \surnamestart Gurevich\surnameend} \&
  \bibinfo{author}{James~K. \surnamestart Huggins\surnameend}
  (\bibinfo{year}{1992}): \emph{\bibinfo{title}{The Semantics of the {C}
  Programming Language}}.
\newblock In \bibinfo{editor}{Egon \surnamestart B{\"o}rger\surnameend},
  \bibinfo{editor}{Gerhard \surnamestart J{\"a}ger\surnameend},
  \bibinfo{editor}{Hans~Kleine \surnamestart B{\"u}ning\surnameend},
  \bibinfo{editor}{Simone \surnamestart Martini\surnameend} \&
  \bibinfo{editor}{Michael~M. \surnamestart Richter\surnameend}, editors: {\sl
  \bibinfo{booktitle}{Computer Science Logic, CSL~'92}}, {\sl
  \bibinfo{series}{LNCS}} \bibinfo{volume}{702}, \bibinfo{publisher}{Springer},
  pp. \bibinfo{pages}{274--308}, \doi{10.1007/3-540-56992-8\_17}.

\bibitemdeclare{incollection}{vfiasco-types}
\bibitem{vfiasco-types}
\bibinfo{author}{M.~\surnamestart Hohmuth\surnameend} \&
  \bibinfo{author}{H.~\surnamestart Tews\surnameend} (\bibinfo{year}{2003}):
  \emph{\bibinfo{title}{The Semantics of {C++} Data Types: Towards Verifying
  low-level System Components}}.
\newblock In \bibinfo{editor}{David \surnamestart Basin\surnameend} \&
  \bibinfo{editor}{Burkhart \surnamestart Wolff\surnameend}, editors: {\sl
  \bibinfo{booktitle}{Theorem Proving in Higher Order Logics, 16th
  International Conference, TPHOLs 2003. Emerging Trends Proceedings}},
  \bibinfo{publisher}{Universit\"at Freiburg}, pp. \bibinfo{pages}{127--144}.

\bibitemdeclare{inproceedings}{vfiasco}
\bibitem{vfiasco}
\bibinfo{author}{M.~\surnamestart Hohmuth\surnameend} \&
  \bibinfo{author}{H.~\surnamestart Tews\surnameend} (\bibinfo{year}{2005}):
  \emph{\bibinfo{title}{The {VFiasco} approach for a verified operating
  system}}.
\newblock In \bibinfo{editor}{Andreas \surnamestart Gal\surnameend} \&
  \bibinfo{editor}{Christian~W. \surnamestart Probst\surnameend}, editors: {\sl
  \bibinfo{booktitle}{Proc.\ 2nd ECOOP Workshop on Programming Languages and
  Operating Systems (ECOOP-PLOS 2005)}}.

\bibitemdeclare{manual}{c11}
\bibitem{c11}
\bibinfo{organization}{ISO/IEC JTC1/SC22/WG14 C Standards Committee}
  (\bibinfo{year}{2011}): \emph{\bibinfo{title}{{\it Programming
  Languages---{C}}}}.
\newblock \bibinfo{note}{{ISO/IEC} 9899:2011}.

\bibitemdeclare{manual}{cpp11}
\bibitem{cpp11}
\bibinfo{organization}{ISO/IEC JTC1/SC22/WG21 C++ Standards Committee}
  (\bibinfo{year}{2011}): \emph{\bibinfo{title}{{\it Programming
  Languages---{C++}}}}.
\newblock \bibinfo{note}{{ISO/IEC} 14882:2011}.

\bibitemdeclare{article}{l4_verified}
\bibitem{l4_verified}
\bibinfo{author}{Gerwin \surnamestart Klein\surnameend} et~al.
  (\bibinfo{year}{2010}): \emph{\bibinfo{title}{{seL4}: formal verification of
  an operating-system kernel}}.
\newblock {\sl \bibinfo{journal}{Commun. ACM}}
  \bibinfo{volume}{53}(\bibinfo{number}{6}), pp. \bibinfo{pages}{107--115},
  \doi{10.1145/1743546.1743574}.

\bibitemdeclare{inproceedings}{Loginov-run-time-check}
\bibitem{Loginov-run-time-check}
\bibinfo{author}{Alexey \surnamestart Loginov\surnameend},
  \bibinfo{author}{Suan~Hsi \surnamestart Yong\surnameend},
  \bibinfo{author}{Susan \surnamestart Horwitz\surnameend} \&
  \bibinfo{author}{Thomas~W. \surnamestart Reps\surnameend}
  (\bibinfo{year}{2001}): \emph{\bibinfo{title}{Debugging via Run-Time Type
  Checking}}.
\newblock In \bibinfo{editor}{Heinrich \surnamestart Hu{\ss}mann\surnameend},
  editor: {\sl \bibinfo{booktitle}{Fundamental Approaches to Software
  Engineering, FASE 2001}}, {\sl \bibinfo{series}{LNCS}}
  \bibinfo{volume}{2029}, \bibinfo{publisher}{Springer}, pp.
  \bibinfo{pages}{217--232}, \doi{10.1007/3-540-45314-8\_16}.

\bibitemdeclare{techreport}{Norrish:CPP2008}
\bibitem{Norrish:CPP2008}
\bibinfo{author}{Michael \surnamestart Norrish\surnameend}
  (\bibinfo{year}{2008}): \emph{\bibinfo{title}{A Formal Semantics for {C++}}}.
\newblock \bibinfo{type}{Technical Report}, \bibinfo{institution}{NICTA}.
\newblock \bibinfo{note}{Available from
  \url{http://nicta.com.au/people/norrishm/attachments/bibliographies_and_pape%
rs/C-TR.pdf}. Retrieved June~15, 2012.}

\bibitemdeclare{inproceedings}{pvs2008}
\bibitem{pvs2008}
\bibinfo{author}{Sam \surnamestart Owre\surnameend} \&
  \bibinfo{author}{Natarajan \surnamestart Shankar\surnameend}
  (\bibinfo{year}{2008}): \emph{\bibinfo{title}{A Brief Overview of {PVS}}}.
\newblock In \bibinfo{editor}{Otmane~A\"{\i}t \surnamestart
  Mohamed\surnameend}, \bibinfo{editor}{C{\'e}sar \surnamestart
  Mu{\~n}oz\surnameend} \& \bibinfo{editor}{Sofi{\`e}ne \surnamestart
  Tahar\surnameend}, editors: {\sl \bibinfo{booktitle}{Theorem Proving in
  Higher Order Logics, 21st International Conference, TPHOLs 2008}}, {\sl
  \bibinfo{series}{LNCS}} \bibinfo{volume}{5170},
  \bibinfo{publisher}{Springer}, pp. \bibinfo{pages}{22--27},
  \doi{10.1007/978-3-540-71067-7\_5}.

\bibitemdeclare{inproceedings}{bitc}
\bibitem{bitc}
\bibinfo{author}{Jonathan \surnamestart Shapiro\surnameend}
  (\bibinfo{year}{2006}): \emph{\bibinfo{title}{Programming language challenges
  in systems codes: why systems programmers still use {C}, and what to do about
  it}}.
\newblock In \bibinfo{editor}{Christian~W. \surnamestart Probst\surnameend},
  editor: {\sl \bibinfo{booktitle}{Proc.\ 3rd Workshop on Programming Languages
  and Operating Systems (PLOS 2006)}}, \bibinfo{publisher}{ACM},
  p.~\bibinfo{pages}{9}.

\bibitemdeclare{article}{tews:memory_peculiarities}
\bibitem{tews:memory_peculiarities}
\bibinfo{author}{Hendrik \surnamestart Tews\surnameend},
  \bibinfo{author}{Marcus \surnamestart V{\"o}lp\surnameend} \&
  \bibinfo{author}{Tjark \surnamestart Weber\surnameend}
  (\bibinfo{year}{2009}): \emph{\bibinfo{title}{Formal Memory Models for the
  Verification of Low-Level Operating-System Code}}.
\newblock {\sl \bibinfo{journal}{Journal of Automated Reasoning: Special Issue
  on Operating Systems Verification}}
  \bibinfo{volume}{42}(\bibinfo{number}{2--4}), pp. \bibinfo{pages}{189--227}.

\bibitemdeclare{techreport}{tews08verification}
\bibitem{tews08verification}
\bibinfo{author}{Hendrik \surnamestart Tews\surnameend}, \bibinfo{author}{Tjark
  \surnamestart Weber\surnameend}, \bibinfo{author}{Marcus \surnamestart
  V{\"o}lp\surnameend}, \bibinfo{author}{Erik \surnamestart Poll\surnameend},
  \bibinfo{author}{Marko~van \surnamestart Eekelen\surnameend} \&
  \bibinfo{author}{Peter~van \surnamestart Rossum\surnameend}
  (\bibinfo{year}{2008}): \emph{\bibinfo{title}{{Nova} Micro--Hypervisor
  Verification}}.
\newblock \bibinfo{type}{Technical Report} \bibinfo{number}{ICIS--R08012},
  \bibinfo{institution}{Radboud University Nijmegen}.

\bibitemdeclare{inproceedings}{tuch07types}
\bibitem{tuch07types}
\bibinfo{author}{Harvey \surnamestart Tuch\surnameend}, \bibinfo{author}{Gerwin
  \surnamestart Klein\surnameend} \& \bibinfo{author}{Michael \surnamestart
  Norrish\surnameend} (\bibinfo{year}{2007}): \emph{\bibinfo{title}{Types,
  Bytes, and Separation Logic}}.
\newblock In \bibinfo{editor}{Martin \surnamestart Hofmann\surnameend} \&
  \bibinfo{editor}{Matthias \surnamestart Felleisen\surnameend}, editors: {\sl
  \bibinfo{booktitle}{Proc. 34th ACM SIGPLAN-SIGACT Symposium on Principles of
  Programming Languages, POPL 2007}}, \bibinfo{publisher}{ACM}, pp.
  \bibinfo{pages}{97--108}, \doi{10.1145/1190216.1190234}.

\end{thebibliography}

%%%%%%%%%%%%%%%%%%%%%%%%%%%%%%%%%%%%%%%%%%%%%%%%%%%%%%%%%%%%%%%%%%%%%%%%%%%%%%%%

\begin{appendix}
\section{Verifying Safe and not so Safe Kernel Code}
\label{sec:case_study}
%====================================================

\begin{figure}[t]
\begin{center}
%\begin{lstlisting}[style=C++,multicols=2,numbers=left,escapechar=@]
% XXX eptcs problem
\begin{lstlisting}[style=C++,numbers=left,escapechar=@]
@\label{lst:scheduler-line-inherit-trick}@class TCB : public list<TCB> {
public:
  unsigned char priority;
  Msg_Buffer mr;
  ...
  @\label{lst:current}@static inline TCB * current(){
    unsigned long dummy;
    @\label{lst:scheduler-line-asm}@asm volatile ("mov %%esp, %0 \n\t" @\label{lst:scheduler-line-asm-end}@: "+m" (dummy) ::);
    @\label{lst:scheduler-line-cast}@return reinterpret_cast<TCB*> @\label{lst:scheduler-line-shift}@(dummy & ~(1 << L2_TCB_SIZE)); 
  }
};
@\label{lst:scheduler-line-priority-list}@list<TCB> prio_list[Max_Prio];
void 
copy(TCB * dest, unsigned long cnt){
  TCB * src = TCB::current();
  @\label{lst:memcpy_error}@memcpy(src, dest, cnt);
}

...
// preempt current thread
@\label{lst:scheduler-line-current}@TCB * current = TCB::current();
@\label{lst:scheduler-line-read}@unsigned char p = current->priority;

@\label{lst:scheduler-line-insert}@prio_list[p].push_back(current);
...
\end{lstlisting}
\caption{Excerpt of a simple IPC send operation and the kernel code
  that is executed when the current thread is preempted.}
\label{fig:case_study}
\end{center}
\end{figure}

To demonstrate our approach, we have verified termination and hence
type safety of a small piece of microkernel code
(\cref{fig:case_study}). The code in
\crefrange{lst:scheduler-line-current}{lst:scheduler-line-insert} is
part of the scheduler. Upon preemption, it inserts the thread control
block (TCB) of the currently running thread at the back of the
doubly-linked priority list. Also shown but not verified is a
simplified version of the copy routine of the inter-process
communication path. To demonstrate the error checking capabilities, we
modified the call to \lstinline[style=c++inline]{memcpy} in
\cref{lst:memcpy_error} to copy the first
\lstinline[style=c++inline]{cnt} bytes from the sender TCB rather than
from its message buffer \lstinline[style=c++inline]{mr}.

The verification is based on an excerpt of the Robin statement and
expression semantics for \cpp{}~\cite{tews08verification} extended with
the \cpp{} instance of our data-type semantics. We
will first focus on
\crefrange{lst:scheduler-line-current}{lst:scheduler-line-insert}, as
they demonstrate the normal use of our data-type semantics. After
that, we dive into the function
\lstinline[style=c++inline]{TCB::current()}, which extracts the TCB
pointer of the current thread from the processor's kernel stack
pointer, and look at the interplay between
\lstinline[style=c++inline]{list<TCB>::push_back} and the erroneous
call to \lstinline[style=c++inline]{memcpy}.

\subsection{Preempt Current Thread}
\label{sec:preempt_current_thread}
%-----------------------------------

For our example, we use sets $\mathbb{S}^T$ that are rich enough to
fulfill the respective preconditions of
\crefrange{lemma:uninitialized}{lemma:bitwise-copy}, for all used types $T$.
We assume that the compiler inlines the call to
\lstinline[style=c++inline]|push_back| in
\cref{lst:scheduler-line-insert}, which therefore
expands to the usual update of the
\lstinline[style=c++inline]{prev} and \lstinline[style=c++inline]{next}
pointers of double-linked list inserts. By inheriting from
\lstinline[style=c++inline]{list<TCB>},
\lstinline[style=c++inline]{class TCB}-typed objects include these
pointers in their representation. In the course of updating these
pointers, the value of \lstinline[style=c++inline]{current} must be
read to obtain the addresses of these members. This value can only be
obtained by reading the byte list at the address of
\lstinline[style=c++inline]{current} and interpreting it using
$s^{\mathit{TCB}*}.\mathit{from\_byte}$. The assignment to
\lstinline[style=c++inline]{p} in \cref{lst:scheduler-line-read} or
hardware side effects (\eg when reading
\lstinline[style=c++inline]{current->priority}) may modify this byte
list in which case our data-type semantics prevents any
verification.
% Remember, because we have chosen $s^{\mathit{TCB}*}$
% arbitrarily from $\mathbb{S}^{\mathit{TCB}*}$, the result of
% $\mathit{from\_byte}$ on this modified byte list must be defined for
% all $s \in \mathbb{S}^{\mathit{TCB}*}$ that allow TCB pointers at the
% address \lstinline[style=c++inline]{current} but then
% $\mathbb{S}^{\mathit{TCB}*}$ would not be type sensitive.
We therefore make the (sensible) assumption that
the
objects at \lstinline[style=c++inline]{current} and
\lstinline[style=c++inline]{p} are allocated at disjoint address
regions, which are not changed by any side effect.\footnote{
  These two assumptions are only made to simplify the case study.
  In a real verification, the disjointness would be implied by
  the functional correctness of the memory allocator. A suitable
  type-safety invariant would imply that side effects occur only
  in other address regions.
}
Then our rewrite engine
simplifies the typed read of the \lstinline[style=c++inline]|current| pointer
to
\begin{equation*}
s^{\mathit{TCB}*}.\mathit{from\_byte}(s^{\mathit{TCB}*}.\mathit{to\_byte}(v, \mathtt{current}), \mathtt{current})
\end{equation*}
which \cref{eq:semantic:inverse} collapses to $v$ where $v$ is the
result of \lstinline[style=c++inline]{TCB::current()}.

When compared to other approaches, the qualitative difference is that
our approach demands either a proof of disjointness, or
additional assumptions that connect the object representations of
\lstinline[style=c++inline]{unsigned char} and
\lstinline[style=c++inline]{TCB*}. For the same reason,
\cref{lemma:bitwise-copy} demands for a fix of the call to
\lstinline[style=c++inline]{memcpy} because only then the list
invariant can be maintained that running threads are never in the
priority list. An erroneous \lstinline[style=c++inline]{memcpy} of the
characters of not trivially-copyable type
\lstinline[style=c++inline]{list<TCB>} prevents the proof of such an
invariant.

\subsection{TCB::current()}
\label{sec:current}
%---------------------------

The verification of \lstinline[style=c++inline]{TCB::current()}
(\crefrange{lst:current}{lst:scheduler-line-shift}) demonstrates the
inclusion of additional assumptions without restricting $\mathbb{S}^T$
to a point where type sensitivity is no longer given. 

Co-locating the kernel stack next to sufficiently aligned objects is a
common programming pattern in microkernels to quickly retrieve
pointers to these objects. \lstinline[style=c++inline]{TCB::current()}
reads the stack pointer value in \lstinline[style=c++inline]{esp} as
an \lstinline[style=c++inline]{unsigned long}
(\cref{lst:scheduler-line-asm}), rounds it to the object alignment
(\cref{lst:scheduler-line-shift}) and casts it into the respective
pointer type (\cref{lst:scheduler-line-cast}). When verifying the
first operation, it is tempting to fix the encoding of
\lstinline[style=c++inline]{esp} as the four-byte little endian
representation of machine words and to require that the word values of
stack addresses are valid object representations for
\lstinline[style=c++inline]{unsigned}~\lstinline[style=c++inline]{long}. Under
these assumptions, underspecified data-type semantics can detect
modifications that cause the \lstinline[style=c++inline]{esp} value to
point to non-stack addresses. However, modifications that cause the
\lstinline[style=c++inline]{esp} to point to other (possibly
unallocated) stacks remain undetected. An elegant way to circumvent
these problems is to introduce a semantic structure also for the
\lstinline[style=c++inline]{esp} register. The verification is then performed
against a whole family of processors that differ in their choice
of semantics structure for register \lstinline[style=c++inline]|esp|.
%including the one on which the kernel should run.

\end{appendix}

\end{document}